\documentclass[
  12pt,
  abstract=true,
  headings=normal,
  a4paper,
]{scrartcl}

\usepackage{amssymb}

\usepackage{amsfonts,amsmath,amssymb,amsthm}
\usepackage{bm,nicefrac}
\usepackage{tabularx}
\usepackage{listings}
\usepackage{enumerate}
\usepackage{float}
\usepackage{graphicx}
\usepackage{comment}
\usepackage{subcaption}
\usepackage{paralist}
\usepackage[ruled,lined,algosection]{algorithm2e}
\usepackage[english]{babel}
\usepackage{verbatim}

\usepackage{amsfonts,amsmath,amssymb,amsthm}
\usepackage{mathtools}
\usepackage{graphicx}
\usepackage{hyperref}
\hypersetup{colorlinks,allcolors=blue}
\usepackage{cancel}
\usepackage{siunitx}
\usepackage{url}
\usepackage{paralist}
\usepackage{adjustbox}
\usepackage{orcidlink}

\newtheorem{lemma}{Lemma}

\newcommand*{\vb}[1]{\mathbf{#1}}
\newcommand*{\vu}[1]{\hat{\mathbf{#1}}}

\newcommand{\myref}[2]{\hyperref[#2]{#1~\ref*{#2}}}

\numberwithin{equation}{section}
\usepackage{matlab-prettifier}

\def\clap#1{\hbox to 0pt{\hss#1\hss}}

\numberwithin{equation}{section}

\usepackage[affil-it]{authblk}

\usepackage{xpatch}
\xpatchcmd{\author}{\relax#1\relax}{\relax\detokenize{#1}\relax}{}{}
\author[\empty]{Yi-Kai Kan\orcidlink{0000-0001-5567-0579}\textsuperscript{1,}\thanks{ykan@bnl.gov  }}
\author[2]{Ji Qiang}
\affil[1]{Brookhaven National Laboratory, Upton, New York, USA}
\affil[2]{Lawrence Berkeley National Laboratory, Berkeley, California, USA}

\addtokomafont{title}{\large}
\begin{document}
\title{On the Electromagnetic Field of a Focusing Charged Particle Beam and Its Two-Dimensional Approximation}
\date{}

\maketitle

\begin{abstract}
The analytical expressions for the electromagnetic potential generated from a focusing charged particle beam are indispensable in various beam physics problems. In this article, we review the theory in detail and point out the necessary assumptions made in the derivation. 
\end{abstract}



\section{Introduction}
The theoretical understanding of the electromagnetic potential generated from a charged particle beam is crucial to many studies, \emph{e.g.}, space-charge and beam-beam interaction. While the physics problem can be fully described by the inhomogeneous electromagnetic wave equation, it is generally not easy to solve when a charge source with a complicated distribution is encountered. Therefore, the quasistatic approximation (model) is widely applied in the beam physics community; the electrostatic field of the particle beam is first solved in the beam's rest-frame, and the corresponding field in the lab-frame is then derived by using the Lorentz transformation. Under the quasistatic approximation, the analytical expression for the electromagnetic field of a rigid Gaussian bunch can be derived and can be found in many literature~\cite{kheifets1976potential,stupakov2018classical,chao2022special}. 

However, in the beam-beam studies, a focusing beam usually needs to be considered, which leads to some critical phenomena, like the longitudinal beam-beam effects~\cite{danilov1992longitudinal,hogan1993longitudinal}. The expression of the electromagnetic potential of a focusing beam is usually given by replacing the constant transverse beam sizes in the three-dimensional (3D) potential of a rigid beam with an $s$-dependent one~\cite{hirata1993symplectic,zhou2022formulae}. However, the legitimation of this ``na\"ive'' generalization has not been rigorously discussed. 

On the other hand, in the beam-beam community, a two-dimensional (2D) formulation is often used for modeling the potential of a particle beam rather than solving the 3D inhomogeneous wave equation; a particle beam is chunked into thin slices along the longitudinal direction, and the field from each slice is calculated by solving the 2D Poisson equation. The 2D potential for a single slice in a Gaussian beam and its closed-form formula have been derived and applied in the literature~\cite{hirata1993symplectic,bassetti1980closed}. This 2D potential should be derivable from the 3D potential under certain approximations. While the parameters characterizing the 2D approximation were discussed in some literature~\cite{kim2011beam-beam, zimmermann1997longitudinal}, how the 2D potential can be derived from the 3D potential was not rigorously demonstrated.
 
In this article, we try to provide a rigorous derivation of the 3D potential from a focusing beam and its 2D approximation. We first formalize the notion of the quasistatic approximation through a formulation based on the inhomogeneous wave equation. After that, based on the discussed quasistatic model, we derive the 3D potential of a Gaussian-distributed focusing beam and its 2D approximation.

\section{The Quasistatic Model}\label{sec:quasistatic}
 The quasistatic model is a commonly used approximation to solve the electromagnetic field of a moving charged particle beam in the beam physics community. One common approach to derive the quasistatic model is based on the Lorentz transformation; the physics problem is first formulated as electrostatic in the beam's frame, and then the full set of governing equations is transformed into the lab frame~\cite{stupakov2018classical}. Despite its popularity in the community, what the quasistatic model actually approximates was seldom discussed. Thus, we first try to formalize the idea of a quasistatic model based on the formulation originally appeared in~\cite{qiang2004parallel} and provide the necessary theoretical reasoning. The tool developed in this section will be a foundation for the discussion in the next section.

The free-space inhomogeneous wave equations for the electric scalar potential $\phi$ and the magnetic vector potential $\vb{A}$ can be expressed as~\cite{jackson1999classical}
\begin{align}
    \nabla^2 \phi(x,y,s,t) - \frac{1}{c^2_0}\frac{\partial^2}{\partial t^2}\phi(x,y,s,t) &= -\frac{\rho(x,y,s,t)}{\varepsilon_0}, \label{eq:wave_phi} \\
    \nabla^2 \vb{A}(x,y,s,t) - \frac{1}{c^2_0} \frac{\partial^2}{\partial t^2}\vb{A}(x,y,s,t) &= -\mu_0 \vb{J}(x,y,s,t) \label{eq:wave_a},
\end{align}
where $\varepsilon_{0}$ is the vacuum permittivity, $\mu_0$ is the vacuum permeability, and $c_0$ is the speed of light in vacuum. Here, the charge density $\rho$ and the current density $\vb{J}$ need to satisfy the continuity equation
\begin{equation}\label{eq:continuity_eq}
    \nabla\cdot\vb{J}(x,y,s,t)  + \frac{\partial }{\partial t}\rho(x,y,s,t) =0,
\end{equation}
and the electromagnetic potentials need to satisfy the Lorenz gauge condition
\begin{equation}
    \nabla\cdot\vb{A} + \frac{1}{c^2_0}\frac{\partial\phi}{\partial t} = 0.
\end{equation}
The quasistatic approximation relies on the following two assumptions:
\begin{compactenum}
    \item Each particle moves almost with the same velocity $v_0$, and the charge and current densities for the particle beam can be modeled as
    \begin{equation*}
        \rho(x,y,s,t) = \rho(x,y,s - v_0 t)\quad\text{and}\quad \vb{J}(x,y,s,t)=v_0 \rho(x,y,s - v_0 t)\vu{s}.
    \end{equation*}
    \item The propagation of the electromagnetic field finishes instantaneously. In this case, given the charge and current densities of the forms $\rho(x,y,s-v_0 t)$ and $\vb{J}(x,y,s-v_0 t)$, the corresponding $\phi$ and $\vb{A}$ satisfying Eq.~\eqref{eq:wave_phi} and Eq.~\eqref{eq:wave_a} can be expressed in the forms $\phi(x,y,s-v_0 t)$ and $\vb{A}(x,y,s-v_0 t)$. The proof is demonstrated in~\myref{Appendix}{sec:quasi_sol_wave_eq}.
\end{compactenum}
Applying the quasistatic approximation, Eq.~\eqref{eq:wave_phi} and Eq.~\eqref{eq:wave_a} become
\begin{align}
    \nabla^2_{\perp}\phi(x,y,z) - \frac{1}{\gamma^2}\frac{\partial^2}{\partial z^2}\phi(x,y,z) &= -\frac{\rho(x,y,z)}{\varepsilon_0}, \label{eq:quasi_wave_phi}\\
    \nabla^2_{\perp} A_s(x,y,z) - \frac{1}{\gamma^2}\frac{\partial^2}{\partial z^2}A_s(x,y,z) &= -\mu_0 v_0 \rho(x,y,z). \label{eq:quasi_wave_a}
\end{align}
Here, the operator $\nabla_{\perp}:=\partial^2/\partial x^2 + \partial^2/\partial y^2$ denotes the two-dimensional Laplacian, and the new variable $z:=s-v_0 t$ stands for the relative position to a reference particle with a trajectory $v_0 t$. After some algebraic manipulations, Eq.~\eqref{eq:quasi_wave_phi} and Eq.~\eqref{eq:quasi_wave_a} reduce to
\begin{equation}\label{eq:simple_relation_phi_a}
\nabla^2_{\perp} (A_s - \frac{v_0}{c^2_0}\phi) + \frac{1}{\gamma^2}\frac{\partial^2}{\partial z^2} (A_s - \frac{v_0}{c^2_0}\phi)=0
\quad\Longrightarrow\quad
A_s=\frac{v_0}{c^2_0}\phi.
\end{equation}
This simple linear relation between $\phi$ and $\vb{A}$ suggests that under the quasistatic approximation, we only need to solve the wave equation for $\phi$, and both electric and magnetic fields can be directly derived from $\phi$. 

\section{The Quasistatic Field of a Focusing Charged Particle Beam}\label{sec:quasi_field_focusing_beam}
Before finding the quasistatic field of a focusing particle beam, we first need to discuss an analytical expression of its charge and current densities. The charge density of a focusing particle beam moving in a constant velocity $v_0$ can be modeled as
\begin{equation}\label{eq:rho_focusing_beam}
    \rho(x,y,s,t)=
    \frac{1}{(2\pi)^{3/2}}\frac{Ne}{\sigma_x(s)\sigma_y(s)\sigma_z}
    \exp\left(-\frac{x^2}{2\sigma^2_x(s)}-\frac{y^2}{2\sigma^2_y(s)}-\frac{(s-v_0 t)^2}{2\sigma^2_z}\right),
\end{equation}
where $N$ is the number of particles and $e$ is the charge of each single particle. 
Here, the horizontal and vertical beam sizes $\sigma_x$ and $\sigma_y$ are defined as
\begin{equation*}\label{eq:trans_bsize}
    \sigma_{i}(s):=\sigma^{*}_i\cdot\left(1+\frac{s^2}{\beta^{*}_{i}}\right)^{1/2}\quad i\in \{x,y\}
\end{equation*}
with $\sigma^{*}_i$ the transverse beam size and $\beta^{*}_i$ the beta function at $s=0$. Throughout this article, we will simply write $\sigma_x$ and $\sigma_y$ without specifying the function argument. It can be checked that Eq.~\eqref{eq:rho_focusing_beam} satisfies $\iiint \rho(x,y,s,t) dxdyds=N e$ at a given time $t$. As each particle moves almost in the velocity $v_0$, it is reasonable to model the longitudinal current density as
\begin{equation}\label{eq:j_focusing_beam}
    J_s(x,y,s,t) = v_0\rho(x,y,s,t).
\end{equation}
Because the transverse beam sizes change with $s$ during the propagation, some transverse current densities ($J_x$ and $J_y$) exist to account for the change in the charge distribution. However, in some scenarios, we may neglect these transverse current densities by claiming $\vert J_x\vert, \vert J_y \vert \ll \vert J_s\vert$. The strength of the current density is proportional to the particle's velocity; and particularly, the particle's transverse velocity is proportional to the rate of change of the beam sizes
\begin{equation*}
    \vert J_i \vert 
    \approx v_0\frac{d\sigma_i}{ds}
    \quad i\in\{x,y\}. 
\end{equation*}
Thus, putting together with $\vert J_s \vert \approx v_0$, we have the following estimation
\begin{equation*}
\left\lvert\frac{J_i}{J_s}\right\rvert
\approx \frac{\sigma^*_i}{\beta^*_i}\frac{\tfrac{s}{\beta^*_i}}{ \sqrt{1+ (\tfrac{s}{\beta^*_i})^2} }
\leq \frac{\sigma^*_i}{\beta^*_i}
= \sqrt{ \frac{\epsilon_i}{\beta^{*}_{i}} }
\quad i\in\{x,y\},
\end{equation*}
where $\epsilon_{i}$ is the transverse emittance of the particle beam. The ratio $\sqrt{\epsilon_i\beta^*_i}$ is the beam divergence and should be small for most particle colliders. In fact, $J_x$ and $J_y$ can be explicitly solved by substituting Eq.~\eqref{eq:rho_focusing_beam} into the continuity equation (Eq.~\eqref{eq:continuity_eq})
\begin{align*}
    &\frac{\partial J_x}{\partial x} + \frac{\partial J_y}{\partial y} + \frac{\partial J_s}{\partial\sigma_x}\frac{\partial\sigma_x}{\partial s} + \frac{\partial J_s}{\partial\sigma_y}\frac{\partial\sigma_y}{\partial s}+ 
    \cancel{ \frac{\partial J_s}{\partial(s-v_0 t)} } = \cancel{ v_0\frac{\rho}{\partial t} } \\
    \Longrightarrow 
    & \frac{\partial J_x}{\partial x} + \frac{\partial J_y}{\partial y} + 
    \left(\frac{x^2}{\sigma^2_x}-1\right)\frac{1}{\sigma_x}\frac{\partial\sigma_x}{\partial s} J_s
    + 
    \left(\frac{y^2}{\sigma^2_y}-1\right)\frac{1}{\sigma_y}\frac{\partial\sigma_y}{\partial s} J_s = 0.
\end{align*}
Hence, the expression for the transverse current densities can be written as
\begin{equation*}
    J_x = \frac{x}{\sigma_x}\frac{\partial\sigma_x}{\partial s}J_s
    \quad\text{and}\quad
    J_y = \frac{y}{\sigma_y}\frac{\partial\sigma_y}{\partial s}J_s.
\end{equation*}

When expressed in the coordinate $(x,y,z,s)$ with $z:=s-v_0 t$, because the charge density still depends explicitly on $s$, the corresponding electromagnetic potential needs to satisfy the wave equations below
\begin{align}
    \nabla^2_{\perp}\phi(x,y,z,s) + \frac{\partial}{\partial s^2}\phi(x,y,z,s) - \frac{1}{\gamma^2}\frac{\partial^2}{\partial z^2}\phi(x,y,z,s) &= -\frac{\rho(x,y,z,s)}{\varepsilon_0}, \label{eq:full_wave_phi_s}\\
    \nabla^2_{\perp} A_s(x,y,z,s) + \frac{\partial}{\partial s^2}A_s(x,y,z,s) - \frac{1}{\gamma^2}\frac{\partial^2}{\partial z^2}A_s(x,y,z,s) &= -\mu_0 v_0 \rho(x,y,z,s). \label{eq:full_wave_a}
\end{align}
By assuming $\phi$ and $A_s$ slowly varying in $s$ (\emph{i.e.}, $\partial^2\phi/\partial s^2 \approx 0$), the simple linear relation between $\phi$ and $A_s$ found in Eq.~\eqref{eq:simple_relation_phi_a} can also be derived 
\begin{equation}\label{eq:simple_relation_full}
    A_s(x,y,z,s)=\frac{v_0}{c^2_0}\phi(x,y,z,s).
\end{equation}
Therefore, it suffices just to solve Eq.~\eqref{eq:full_wave_phi_s}, and the solution can be found by one of the methods used in~\cite{kheifets1976potential,chao2022special}
\begin{equation}\label{eq:phi3d}
    \phi=
    \frac{1}{4\pi\varepsilon_0}\frac{Ne}{(8\pi)^{1/2}}
    \int^{\infty}_{0}
    \frac{
        \exp\left(-\frac{x^2}{2(\sigma^2_x + q)}-\frac{y^2}{2(\sigma^2_y + q)}-\frac{z^2}{2(\sigma^2_z + q/\gamma^2)}\right)
    }{
        (\sigma^2_x + q)^{1/2}(\sigma^2_y + q)^{1/2}(\sigma^2_z + q/\gamma^2)^{1/2}
    }dq.
\end{equation}
Eq.~\eqref{eq:phi3d} is almost the same as the solution for a rigid bunch except that $\sigma_x$ and $\sigma_y$ are now functions of $s$. Applying Eq.~\eqref{eq:simple_relation_full}, the calculation of the electric field can be simplified:
\begin{equation}\label{eq:efields_3d}
    \vb{E}=-\nabla \phi - \frac{\partial \vb{A}}{\partial t}
    \quad\Longrightarrow\quad
    E_x = -\frac{\partial\phi}{\partial x}, 
    E_y = -\frac{\partial\phi}{\partial y},\text{ and } 
    E_s = 
    \overbracket{ -\frac{\partial\phi}{\partial s} }^{ =: E_{ss} }
    \underbracket{ -\frac{1}{\gamma^2}\frac{\partial\phi}{\partial z} }_{ =: E_{sz} },
\end{equation}
where the term $E_{ss}$ is caused by the beam focusing during the propagation, and the term $E_{sz}$ is due to the density variation of particles in $z$-coordinate.

\section{The 2D Approximation}\label{sec:2d_approximation}
To derive the 2D approximation of Eq.~\eqref{eq:phi3d}, we introduce a change of variable \\${q:=\min(\sigma_x,\sigma_y)^2\cdot w/(1-w)}$. For an elliptical beam with $\sigma_y \leq \sigma_x$, we choose \\ ${q=\sigma^2_y\cdot w/(1-w)}$, and Eq.~\eqref{eq:phi3d} becomes
\begin{equation}\label{eq:phi3d_w}
\begin{split}
    &\phi 
    = \frac{1}{4\pi\varepsilon_0}\frac{Ne}{(8\pi)^{1/2}}
    \frac{\sigma_y}{\sigma_x\sigma_z}\times  \\
    &\int^{1}_{0}
    \underbrace{
        \dfrac{\exp\left(-\frac{(1-w)x^2}{2\sigma^2_x(1 - w + A w)}-\frac{(1-w)y^2}{2\sigma^2_y}-\tfrac{(1-w)\gamma^2z^2}{2\gamma^2\sigma^2_z(1 - w + \epsilon w)}\right)}{(1-w)^{1/2}(1-w+ A w)^{1/2}(1-w+\epsilon w)^{1/2}}
    }_{=:\psi_{\epsilon}(w)}dw
\end{split}
\end{equation}
with $\epsilon:=\sigma^2_y/(\gamma^2\sigma^2_z)$ and $A:=\sigma^2_y/\sigma^2_x$. Now, we want to calculate the electric field $E_{x}$, $E_{y}$, $E_{ss}$ and $E_{sz}$ in the limit $\epsilon \to 0$.   
In the calculation for $E_x$ and $E_{y}$, we can interchange the limit and integration 
\begin{align*}
\lim_{\epsilon\to 0} E_x &=\text{const}\cdot\lim_{\epsilon\to 0}\int^{1}_{0} \dfrac{\partial}{\partial x}\psi_{\epsilon}(w)dw 
     = \text{const}\cdot\int^{1}_{0} \lim_{\epsilon\to 0}\dfrac{\partial}{\partial x}\psi_{\epsilon}(w)dw, \\
\lim_{\epsilon\to 0} E_y &= \text{const}\cdot\lim_{\epsilon\to 0}\int^{1}_{0} \dfrac{\partial}{\partial y}\psi_{\epsilon}(w)dw
     = \text{const}\cdot\int^{1}_{0} \lim_{\epsilon\to 0}\dfrac{\partial}{\partial y}\psi_{\epsilon}(w)dw, 
\end{align*}
because $\vert\partial \psi_\epsilon(w)/\partial x\vert$ and $\vert\partial \psi_\epsilon(w)/\partial y\vert$ are smaller than some integrable functions and the dominated convergence theorem can be ultilized~\cite{axler2020measure}. Therefore, we can further write out
\begin{equation}\label{eq:exey_w}
\begin{split}
    \lim_{\epsilon\to 0} E_{x} &=
    \text{const}\cdot \exp\left(-\frac{z^2}{2\sigma^2_{z}}\right) \frac{x}{\sigma^2_x}
    \int^{1}_{0} \dfrac{\exp\left(-\frac{(1-w)x^2}{2\sigma^2_x(1 - w + A w)}-\frac{(1-w)y^2}{2\sigma^2_y}\right)}{(1-w+ A w)^{3/2}}dw,
    \\
    \lim_{\epsilon\to 0} E_{y} &= \text{const}\cdot \exp\left(-\frac{z^2}{2\sigma^2_{z}}\right) \frac{y}{\sigma^2_y}
    \int^{1}_{0} \dfrac{\exp\left(-\frac{(1-w)x^2}{2\sigma^2_x(1 - w + A w)}-\frac{(1-w)y^2}{2\sigma^2_y}\right)}{(1-w+ A w)^{1/2}}dw.
\end{split}
\end{equation}
Using a substitution of $w=q/(q+\sigma^2_y)$ and some algebraic manipulation, Eqs.~\eqref{eq:exey_w} can be expressed in the form
\begin{equation}\label{eq:exey_2d}
    E^{\mathrm{2D}}_x := \lim_{\epsilon\to 0} E_{x} =-\frac{\partial\phi^{\textrm{2D}}}{\partial x}
    \quad\text{and}\quad
    E^{\mathrm{2D}}_y := \lim_{\epsilon\to 0} E_{y} =-\frac{\partial\phi^{\textrm{2D}}}{\partial y}
\end{equation}
Here, we define a 2D scalar potential
\begin{equation}\label{eq:phi2d}
    \phi^{\mathrm{2D}}:=\frac{1}{4\pi\varepsilon_0}\frac{e \cdot \rho_{\parallel}(z)}{2}
    \int^{\infty}_{0}
    \frac{\exp\left(-\frac{x^2}{2(\sigma^2_x + q)}-\frac{y^2}{2(\sigma^2_y + q)}\right)}{(\sigma^2_x + q)^{1/2}(\sigma^2_y + q)^{1/2}}dq,
\end{equation}
with
\begin{equation}
    \rho_{\parallel}(z):=\frac{1}{(2\pi)^{1/2}\sigma_z}\exp\left(-\frac{z^2}{2\sigma^2_z}\right)
\end{equation}
the normalized longitudinal particle density. Actually, Eq.~\eqref{eq:phi2d} is the solution of the two-dimensional Poisson equation
\begin{equation}
    \nabla^2_{\perp}\phi(x,y,z) = -\frac{\rho(x,y,z)}{\varepsilon_0},
\end{equation}
which is a simplification of Eq.~\eqref{eq:quasi_wave_phi} by neglecting the term $\frac{1}{\gamma^2}\frac{\partial^2 \phi}{\partial s^2}$ and is a theoretical foundation of some beam-beam studies, \emph{e.g.}, the strong-strong simulation model~\cite{qiang2004parallel}. For the calculation of $\lim_{\epsilon\to 0}E_{ss}$, we can first observe from Eq.~\eqref{eq:phi3d} that~\cite{zimmermann1997longitudinal}
\begin{equation}
    E_{ss} 
    = -\frac{\partial\phi}{\partial s}
    = -\frac{\partial(\sigma^2_x)}{\partial s}\frac{\partial\phi}{\partial(\sigma^2_x)} -\frac{\partial(\sigma^2_y)}{\partial s}\frac{\partial\phi}{\partial(\sigma^2_y)}
    = -\frac{1}{2}\frac{\partial(\sigma^2_x)}{\partial s}\frac{\partial^2\phi}{\partial x^2} - \frac{1}{2}\frac{\partial(\sigma^2_y)}{\partial s}\frac{\partial^2\phi}{\partial y^2}.
\end{equation}
Hence, we can apply the same machinery in the derivation of Eqs.~\eqref{eq:exey_2d} and eventually get the result 
\begin{equation}
    E^{\mathrm{2D}}_{ss}:=\lim_{\epsilon\to 0}E_{ss} = -\frac{\partial\phi^{\mathrm{2D}}}{\partial s}.
\end{equation}
A set of close-form expressions of $E^{2D}_{x}$, $E^{2D}_{y}$ and $E^{2D}_{ss}$ can be further dervied~\cite{hirata1993symplectic, bassetti1980closed}, and has been widely used in many beam-beam studies~\cite{hirata1993symplectic}.

Before going through a detailed derivation of $\lim_{\epsilon\to 0}E_{sz}$, we first try to guesstimate the answer. The limit $\epsilon=\sigma^{2}_{y}/(\gamma^2\sigma^2_z)\to 0$ implies a particle beam with a very large Lorentz factor or a very long bunch length. We know the electric field lines emanating from a charged particle get compressed in the transverse direction due to the Lorentz contraction~\cite{jackson1999classical}. In the highly relativistic limit $\gamma\to\infty$, the longitudinal field of each single particle in a bunch approaches zero, and the same thing also holds for the field of the whole bunch as it is just the superposition of single-particle fields. To discuss the case of an extremely long bunch, we consider a particle beam with a Gaussian-distributed density $\rho$ and a bunch length $\sigma_z$. Given an arbitrary location $z=z_0$ and a length $l>0$, we have $\rho(z=z_0 - l) \approx \rho(z=z_0 + l)$ as $\sigma_z\to\infty$; and hence, the longitudinal electric fields generated from $\rho(z=z_0 - l)$ and $\rho(z=z_0 + l)$ cancel at $z=z_0$ for all $l$. 

In deriving $\lim_{\epsilon\to 0}E_{sz}$, if limit and integration are arbitrarily interchanged, we may get a result
\begin{align}
    \lim_{\epsilon\to 0}E_{sz} 
    =& \text{const}\cdot\lim_{\epsilon\to 0}\int^{1}_{0} \dfrac{\partial}{\partial z}\psi_{\epsilon}(w)dw \label{eq:sz_interchange_lim_int}\\
    \stackrel{?}{=}& \text{const}\cdot\int^{1}_{0} \lim_{\epsilon\to 0}\dfrac{\partial}{\partial z}\psi_{\epsilon}(w)dw  \\
    =& \text{const}\cdot \exp\left(-\frac{z^2}{2\sigma^2_{z}}\right) \frac{z}{\sigma^2_z}
    \int^{1}_{0} \dfrac{\exp\left(-\frac{(1-w)x^2}{2\sigma^2_x(1 - w + A w)}-\frac{(1-w)y^2}{2\sigma^2_y}\right)}{(1-w+ A w)^{1/2}(1-w)}dw \label{eq:esz_div}\\
    =& -\frac{1}{\gamma^2}\frac{\partial\phi^{\mathrm{2D}} }{\partial z}\quad \text{(substitute with } w=q/(q+\sigma^2_y)\text{)}.\label{eq:esz_2d_wrong}
\end{align}
However, Eq.~\eqref{eq:esz_2d_wrong} diverges because the integrand in Eq.~\eqref{eq:esz_div} is bigger than a function $f(w):=\exp(-\tfrac{x^2}{2\sigma^2_x}-\tfrac{y^2}{2\sigma^2_y})/(1-w)$ on the interval $[0,1]$, and the integral $\int^{1}_{0}f(w)dw$ diverges. This means that the result of Eq.~\eqref{eq:esz_2d_wrong} and our previous guestimation contradict each other. This contradiction arises because we cannot find an integrable function that is bigger than $\vert\partial \psi(w)_{\epsilon}/\partial z\vert$; and hence, we are not allowed to interchange integration and limit in Eq.~\eqref{eq:sz_interchange_lim_int}. To find out the right answer, we first write down $E_{sz}$ explicitly
\begin{equation}\label{eq:phi_2d_w}
\begin{split}
    E_{sz} 
    =& \frac{1}{4\pi\varepsilon_0}\frac{Ne}{(8\pi)^{1/2}}\frac{\sigma_y}{\gamma^2\sigma_x\sigma^2_z}\frac{z}{\sigma_z} \times \\
    &
    \underbrace{
    \int^{1}_{0}
    \dfrac{
        (1-w)^{1/2}\exp\left(-\frac{(1-w)x^2}{2\sigma^2_x(1 - w + A w)}-\frac{(1-w)y^2}{2\sigma^2_y}-\tfrac{(1-w)\gamma^2z^2}{2\gamma^2\sigma^2_z(1 - w + \epsilon w)}\right)
    }{
        (1-w+ A w)^{1/2}(1-w+\epsilon w)^{3/2}
    }dw.
    }_{=:I_{\epsilon}}
\end{split}
\end{equation}
With a new variable $u:=(1-w)/\epsilon$ defined, an upper bound of $I_{\epsilon}$ can be derived
\begin{align}
    I_{\epsilon} \leq &\int^{1/\epsilon}_{0} \frac{u^{1/2}}{ (A+ (1 - A) \epsilon u )^{1/2} \cdot(1+(1-\epsilon)u)^{3/2}}dw \\
      \leq &\frac{1}{(1-A)^{1/2}\cdot\epsilon^{1/2}}\int^{1/\epsilon}_{0} \frac{1}{ (1+(1-\epsilon)u)^{3/2}}du \label{eq:I_2}\\
          =& \frac{2}{(1 - A)^{1/2}\cdot \epsilon^{1/2}} \frac{1-\epsilon^{1/2}}{1-\epsilon} \label{eq:I_3}.
\end{align}
In deriving Eq.~\eqref{eq:I_2}, we apply an inequality $1/(A+ (1 - A) \epsilon u) \leq 1/( (1 - A) \epsilon u )$, which is true for $u > 0$ because $\sigma_y/\sigma_x \leq 1 $ was assumed at the beginning, and $(1-A) \geq 0 $. Combining Eq.~\eqref{eq:phi_2d_w} and Eq.~\eqref{eq:I_3}, we conclude that the magnitude of $E_{sz}$ is bounded by
\begin{equation*}
    \vert E_{sz} \vert 
    \leq
    \frac{1}{4\pi\varepsilon_0}\frac{Ne}{(8\pi)^{1/2}}\frac{2}{\sigma_x\sigma_y}
    \frac{\vert z\vert}{\sigma_z}
    \frac{1}{1-\sigma^2_y/\sigma^2_x}
    \frac{\epsilon^{1/2}}{1+\epsilon^{1/2}},
\end{equation*}
and $\vert E_{sz}\vert \to 0$ as $\epsilon\to 0$. Therefore, we have $E^{2D}_{sz}:=\lim_{\epsilon\to 0} E_{sz} = 0$.




\section*{Acknowledgement}
This work is supported by the U.S. DOE under contract NO. DE-AC02-05CH11231.

\appendix
\section{A Property on the Solution of the Inhomogeneous Wave Equation}\label{sec:quasi_sol_wave_eq}
\begin{lemma}\label{thm:quasi_sol_wave_eq}
Assume the electromagnetic wave follows instantaneous propagation; in other words, every signal arrives at the observer right after its generation from the source without any retardation. Given a time-dependent charge source of the forms $\rho(x,y,s-v_0 t)$ and $\vb{J}(x,y,s-v_0 t)$, the solutions to the inhomogeneous wave equations for the scalar and vector potentials can be expressed in the forms $\phi(x,y,s-v_0 t)$ and $\vb{A}(x,y,s-v_0 t)$.
\end{lemma}
\begin{proof}
We only show the case for $\phi$, and the case for $\vb{A}$ can be done via the same machinery. The retarded solution of the inhomogeneous wave equation for $\phi$ is~\cite{jackson1999classical}
\begin{equation}\label{eq:phi_retarded_form}
\phi(x,y,s,t) =\frac{1}{4\pi\varepsilon_0}\iiint g(x-x',y-y',s-s')\cdot\rho(x',y',s',t')dx'dy'ds' \\ 
\end{equation}
with a Green's function 
\begin{equation}
    g(x-x',y-y',s-s') = 1/\sqrt{(x-x')^2 + (y-y')^2 + (s-s')^2}.
\end{equation}
Here, $t'$ satisfies the retardation condition 
\begin{equation}
    t'=t - \sqrt{(x-x')^2 + (y-y')^2 + (s-s')^2}/c_0.
\end{equation}
Substituting the source term in Eq.~\eqref{eq:phi_retarded_form} with $\rho(x',y',s'-v_0 t')$ and applying the assumption $t'=t$, we can prove the conclusion
\begin{align*}
    \phi(x,y,s,t) 
    &=\frac{1}{4\pi\varepsilon_0}\iiint g(x-x',y-y',s-s')\cdot\rho(x',y',s'-v_0 t)dx'dy'ds' \\
    &=\frac{1}{4\pi\varepsilon_0}\iiint g(x-x',y-y',s-v_0 t -s'')\cdot\rho(x',y',s'')dx'dy'ds'' \quad (s'':=s'-v_0 t) \\
    &=\frac{1}{4\pi\varepsilon_0}\iiint g(x-x',y-y', (s-v_0 t) -s'')\cdot\rho(x',y',s'')dx'dy'ds''  \\
    &=:f(x,y,s-v_0 t).
\end{align*}
\end{proof}


\bibliographystyle{elsarticle-num}
\bibliography{references}







\end{document}